\documentclass[letterpaper, 10pt, conference]{ieeeconf}
\IEEEoverridecommandlockouts
\usepackage[OT1]{fontenc}
\usepackage{graphicx, bm, cite,array}
\usepackage{amssymb, dsfont, color, epstopdf}
\usepackage[cmex10]{amsmath}
\usepackage{enumerate}
\usepackage{apptools}
\usepackage{multirow}
\usepackage{booktabs}
\usepackage{threeparttable}
\usepackage{subeqnarray}
\usepackage{cases, stmaryrd}
\usepackage{stackengine}

\newcommand{\overbar}[1]{\mkern 1.5mu\overline{\mkern-1.5mu#1\mkern-1.5mu}\mkern 1.5mu}

\newtheorem{theorem}{Theorem}

\newtheorem{lemma}{Lemma}

\newtheorem{definition}{Definition}

\newcommand{\Id}{\textmd{Id}}
\newcommand{\inte}{\textrm{Int}}

\newcommand{\rclf}{\textsf{v}}
\newcommand{\rcbf}{\textsf{b}}
\newcommand{\rclbf}{\textsf{w}}
\newcommand{\cclf}{\textrm{c}}

\newcommand{\dly}{\text{d}}

\begin{document}

\title{\LARGE \bf Razumikhin-type Control Lyapunov and Barrier Functions for Time-Delay Systems
\thanks{This work was supported by the European Research Council through the European Union's Horizon 2020 Research and Innovation Programme under Grant 864017--L2C, the Walloon Region and the Innoviris Foundation.}
}

\author{Wei~Ren
\thanks{W. Ren is with ICTEAM Institute, Universit\'{e} catholique de Louvain, 1348 Louvain-la-Neuve, Belgium. Email: \texttt{w.ren@uclouvain.be}.}
}

\maketitle

\begin{abstract}
This paper studies the stabilization and safety problems of nonlinear time-delay systems, where time delays exist in system state and affect the controller design. Following the Razumikhin approach, we propose a novel control Lyapunov-Razumikhin function to facilitate the controller design and to achieve the stabilization objective. To ensure the safety objective, we propose a Razumikhin-type control barrier function for time-delay systems for the first time. Furthermore, the proposed Razumikhin-type control Lyapunov and barrier functions are merged such that the stabilization and safety control design can be combined to address the stabilization and safety simultaneously, which further extends the control design from the delay-free case into the time-delay case. Finally, the proposed approach is illustrated via a numerical example.
\end{abstract}

\section{Introduction}
\label{sec-intro}

With emerging interests in cyber-physical systems \cite{Humayed2017cyber} like networked control systems, robotics and autonomous vehicles, stabilization and safety are two fundamental objectives, which require dynamic systems to achieve the stability (or tracking/synchronization) objective on the one hand and to satisfy safety constraints on the other hand. In particular, the safety objective is placed in the priority position. For safety-critical systems, it is imperative to keep them into the safe set while controlling them. Consequently, the design of the stabilizing controller must comply with state (or input) constraints to ensure the safety. Similar to control Lyapunov functions (CLFs) proposed for the stabilization objective \cite{Sontag1989universal, Jankovic2001control, Pepe2017control}, safety constraints can be specified in terms of a set invariance and verified via control barrier functions (CBFs) \cite{Ames2016control, Romdlony2016stabilization}. CLFs and CBFs have been applied to deal with different objectives for diverse dynamical systems \cite{Ogren2001control, Panagou2015distributed, Lindemann2018control}, and further combined to study the stabilization and safety objectives simultaneously, and the combination can be either implicit \cite{Ngo2005integrator} or explicit \cite{Ames2016control, Jankovic2018robust, Romdlony2016stabilization} via different techniques.

In the fields of engineering, biology and physics, time delays are frequently encountered due to information acquisition and computation for control decisions and executions \cite{Sipahi2011stability}, and may induce many undesired phenomena like oscillation, instability and performance deterioration \cite{Gao2019stability}. For this purpose, both stability and stabilization have been studied extensively in the past decades \cite{Fridman2014tutorial}. Since time delays cause a violation of monotonic decrease conditions of classic Lyapunov functions, two ways to extend the classic Lyapunov-based method are \cite{Ren2019krasovskii}: (i) the Krasovskii approach based on Lyapunov-Krasovskii functionals which are positive definite and whose derivatives are negative definite along system solutions; (ii) the Razumikhin approach based on Lyapunov-Razumikhin functions which are positive definite and whose derivatives are negative definite under the Razumikhin condition \cite{Teel1998connections}. These two approaches have been applied successfully in stability analysis and controller design for time-delay systems \cite{Sipahi2011stability, Ren2018vector, Gao2019stability}. On the other hand, in the aforementioned areas involving time delays, numerous dynamic systems are safety-critical, which results in the need of the safety objective for time-delay systems \cite{Prajna2005methods, Jankovic2018control}. However, most existing results are focused on stability, stabilization and robustness instead of on safety, which motivates us to investigate how to guarantee the safety objective of time-delay systems.

In this letter, we focus on nonlinear systems with state delays and follow the Razumikhin approach to investigate the stabilization and safety problems. First, based the notion of steepest descent feedback \cite{Pepe2014stabilization}, we propose a novel control Lyapunov-Razumikhin function (CLRF) to overcome the verification of the Razumikhin condition and to facilitate the controller design for time-delay systems. With the proposed CLRF, the classic small control property (SCP) is extended to the time-delay case. Therefore, based on the Razumikhin-type CLF and SCP, the closed-from controller is designed to guarantee the stabilization objective. Second, following the similar mechanism and based on the unsafe set (e.g., obstacles and forbidden states) \cite{Prajna2005methods}, the control barrier-Razumikhin function (CBRF) is proposed for time-delay systems for the first time, and further the safety controller is derived explicitly in the closed form. Finally, to achieve the stabilization and safety objectives simultaneously, the proposed CLRF and CBRF are merged to combine both stabilization and safety control design, which results in a novel control design method via the development of the Razumikhin-type control Lyapunov-barrier function (CLBF). In particular, we show how to construct the Razumikhin-type CLBF via the proposed CLRF and CBRF. In conclusion, our main contributions are two-fold: (i) by proposing CLRF and CBRF, both stabilizing and safety controllers are derived explicitly, which extends the Sontag's formula \cite{Sontag1989universal} and the existing results \cite{Ames2016control, Jankovic2018robust} to the time-delay case; (ii) the proposed CLRF and CBRF are merged together such that the stabilizing control and the safety control can be combined in the sense that the stabilization and safety objectives are ensured simultaneously for time-delay systems.

The remainder of this paper is as follows. Preliminaries are presented in Section \ref{sec-nonconsys}. All Razumikhin-type control functions are proposed in Section \ref{sec-Razumikhintype}. Simulation is presented in Section \ref{sec-examples} followed by conclusions and future research in Section \ref{sec-conclusion}. All proofs are located in the Appendix.

\section{Preliminaries}
\label{sec-nonconsys}

Let $\mathbb{R}:=(-\infty, +\infty); \mathbb{R}^{+}:=[0, +\infty); \mathbb{N}:=\{0, 1, \ldots\}$ and $\mathbb{N}^{+}:=\{1, 2, \ldots\}$. $\|x\|$ denotes the Euclidian norm of the vector $x\in\mathbb{R}^{n}$, and $(a,b):=(a^{\top}, b^{\top})^{\top}$ for $a, b\in\mathbb{R}^{n}$. Given a set $\mathbb{C}\subset\mathbb{R}^{n}$, $\partial\mathbb{C}$ is the boundary of $\mathbb{C}$; $\inte(\mathbb{C})$ is the interior of $\mathbb{C}$; $\overbar{\mathbb{C}}$ is the closure of $\mathbb{C}$. Given $\delta>0$ and $\mathbf{x}\in\mathbb{R}^{n}$, an open ball centered at $\mathbf{x}$ with radius $\delta$ is denoted by $\mathbf{B}(\mathbf{x}, \delta):=\{x\in\mathbb{R}^{n}: \|x-\mathbf{x}\|<\delta\}$; $\mathbf{B}(\delta):=\mathbf{B}(0, \delta)$. $\mathcal{C}([a, b], \mathbb{R}^{n})$ denotes the class of piecewise continuous functions mapping $[a, b]$ to $\mathbb{R}^{n}$; $\mathcal{C}(\mathbb{R}^{n}, \mathbb{R}^{p})$ denotes the class of continuously differentiable functions mapping $\mathbb{R}^{n}$ to $\mathbb{R}^{p}$. $\Id$ is the identity function, and $\alpha\circ\beta(v):=\alpha(\beta(v))$ for any $\alpha, \beta\in\mathcal{C}(\mathbb{R}^{n}, \mathbb{R}^{n})$. A function $\alpha: \mathbb{R}^{+}\rightarrow\mathbb{R}^{+}$ is of class $\mathcal{K}$ if it is continuous, $\alpha(0)=0$, and strictly increasing; it is of class $\mathcal{K}_{\infty}$ if it is of class $\mathcal{K}$ and unbounded. A function $\beta: \mathbb{R}^{+}\times\mathbb{R}^{+}\rightarrow\mathbb{R}^{+}$ is of class $\mathcal{KL}$ if $\beta(s, t)$ is of class $\mathcal{K}$ for each fixed $t\geq0$ and $\beta(s, t)\rightarrow0$ as $t\rightarrow0$ for each fixed $s\geq0$. A function $V: \mathbb{R}^{n}\rightarrow\mathbb{R}$ is called \textit{proper} if the sublevel set $\{x\in\mathbb{R}^{n}: V(x)\leq c\}$ is compact for all $c\in\mathbb{R}$, or equivalently, $V$ is radially unbounded.

\subsection{Time-Delay Control Systems}
\label{subsec-nonsystem}

In this letter, we consider nonlinear time-delay control systems with the following dynamics:
\begin{align}
\label{eqn-1}
\begin{aligned}
\dot{x}(t)&=f(x_{t})+g(x_{t})u, &\quad& t>0, \\
x(t)&=\xi(t), &\quad& t\in[-\Delta, 0],
\end{aligned}
\end{align}
where $x\in\mathbb{R}^{n}$ is the system state, $x_{t}=x(t+\theta)\in\mathbb{R}^{n}$ is the time-delay state with $\theta\in[-\Delta, 0]$, and $\Delta>0$ is the upper bound of time delays. The initial state is $\xi\in\mathcal{C}([-\Delta, 0], \mathbb{X}_{0})$ with $\mathbb{X}_{0}\subset\mathbb{R}^{n}$ and $\|\xi\|_{\dly}:=\sup_{\theta\in[-\Delta, 0]}\|\xi(\theta)\|$ being bounded. The control input $u$ takes value from the set $\mathbb{U}\subset\mathbb{R}^{m}$, and the input function is not specified explicitly since it may depend on the current state or/and the time-delay trajectory. Assume that the functionals $f: \mathcal{C}([-\Delta, 0], \mathbb{R}^{n})\rightarrow\mathbb{R}^{n}$ and $g: \mathcal{C}([-\Delta, 0], \mathbb{R}^{n})\rightarrow\mathbb{R}^{n\times m}$ are continuous and locally Lipschitz, which guarantees the existence of the unique solution to the system \eqref{eqn-1}; see \cite[Section 2]{Hale1993introduction}. Also, let $f(0)=0$ and $g(0)=0$, that is, $x(t)\equiv0$ for all $t>0$ is a trivial solution of the system \eqref{eqn-1}. To consider the stabilization problem of the system \eqref{eqn-1}, we assume that the origin is included in the initial set $\mathbb{X}_{0}$.

\begin{definition}[\cite{Sastry2013nonlinear}]
\label{def-1}
Given the control input $u\in\mathbb{U}$, the system \eqref{eqn-1} is \textit{globally asymptotically stable (GAS)} if there exists $\beta\in\mathcal{KL}$ such that $\|x(t)\|\leq\beta(\|\xi\|_{\dly}, t)$ for all $t\geq0$ and all bounded $\xi\in\mathcal{C}([-\Delta, 0], \mathbb{R}^{n})$; and the system \eqref{eqn-1} is \textit{semi-globally asymptotically stable (semi-GAS)} if there exists $\beta\in\mathcal{KL}$ such that $\|x(t)\|\leq\beta(\|\xi\|_{\dly}, t)$ for all $t\geq0$ and all $\xi\in\mathcal{C}([-\Delta, 0], \mathbb{X}_{0})$.
\end{definition}

From Definition \ref{def-1}, the \textit{stabilization control} is to design a feedback controller such that the closed-loop system is GAS. To study the system safety, some notations are defined below. For the system \eqref{eqn-1}, its unsafe set is denoted as an open set $\mathbb{D}\subset\mathbb{R}^{n}$. The system \eqref{eqn-1} is \textit{safe}, if $x(t)\notin\overbar{\mathbb{D}}$ for all $t\geq-\Delta$. Hence, $\mathbb{X}_{0}\cap\mathbb{D}=\varnothing$ is assumed such that $\xi(\theta)\notin\overbar{\mathbb{D}}$ for all $\theta\in[-\Delta, 0]$. The \textit{safety control} is to design a feedback controller to ensure the safety of the closed-loop system.

Since the safety and stabilization objectives of time-delay systems cannot be achieved via classic CLFs and CBFs, our goal is to implement the Razumikhin approach to propose novel types of CLFs and CBFs for the system \eqref{eqn-1}.

\section{Main Results}
\label{sec-Razumikhintype}

In this section, we follow the Razumikhin approach to propose control Lyapunov and barrier functions for time-delay systems. To this end, we first propose a novel control Lyapunov-Razumikhin function for the stabilization objective, then a novel control barrier-Razumikhin function for the safety objective, and finally combine the proposed Razumikhin-type control functions to study the stabilization and safety objectives simultaneously.

\subsection{Control Lyapunov-Razumikhin Functions}
\label{subsec-clrf}

We start with recalling the following control Lyapunov-Razumikhin function from \cite{Jankovic2001control, Pepe2017control}.

\begin{definition}
\label{def-2}
For the system \eqref{eqn-1}, a function $V_{\cclf}\in\mathcal{C}(\mathbb{R}^{n}, \mathbb{R}^{+})$ is called a \textit{control Lyapunov-Razumikhin function (CLRF-I)}, if
\begin{enumerate}[(i)]
  \item there exist $\alpha_{1}, \alpha_{2}\in\mathcal{K}_{\infty}$ such that $\alpha_{1}(\|x\|)\leq V_{\cclf}(x)\leq\alpha_{2}(\|x\|)$ for all $x\in\mathbb{R}^{n}$;

  \item there exist $\gamma_{\cclf}, \rho_{\cclf}\in\mathcal{K}$ with $\rho_{\cclf}(v)>v$ for all $v>0$ such that for all $\phi\in\mathcal{C}([-\Delta, 0], \mathbb{R}^{n})$ with $\phi(0)=x$, if $\rho_{\cclf}(V_{\cclf}(x))\geq\|V_{\cclf}(\phi)\|_{\dly}$, then $\inf_{u\in\mathbb{U}}\{L_{f}V_{\cclf}(\phi)+L_{g}V_{\cclf}(\phi)u\}\leq-\gamma_{\cclf}(V_{\cclf}(x))$,
\end{enumerate}
where $\|V_{\cclf}(\phi)\|_{\dly}:=\sup_{\theta\in[-\Delta, 0]}V_{\cclf}(\phi(\theta))$, $L_{f}V_{\cclf}(\phi):=\frac{\partial V_{\cclf}(x)}{\partial x}f(\phi)$ and $L_{g}V_{\cclf}(\phi):=\frac{\partial V_{\cclf}(x)}{\partial x}g(\phi)$.
\end{definition}

In Definition \ref{def-2}, the effects of time delays are shown via the Razumikhin condition: $\rho_{\cclf}(V_{\cclf}(x))\geq\|V_{\cclf}(\phi)\|_{\dly}$, which can be written equivalently as $\rho_{\cclf}(V_{\cclf}(x))\geq V_{\cclf}(\phi(\theta))$ for all $\theta\in[-\Delta, 0]$; see \cite{Pepe2017control, Ren2018vector}. In particular, if $L_{g}V_{\cclf}(\phi)\equiv0$, then Definition \ref{def-2} is reduced to the one in \cite{Jankovic2001control}.

In the delay-free case \cite{Sontag1989universal}, the existence of the continuous controller is verified via the small control property (SCP). However, due to the Razumikhin condition in the time-delay case, the violation of the Razumikhin condition results in additional difficulties in the controller design, and thus the SCP is not available here. On the other hand, the existing construction of the stabilizing controller is based on the optimization theory \cite{Sepulchre2012constructive} and the trajectory-based approach \cite{Jankovic2001control}, and the closed form of the continuous controller cannot be expressed easily and explicitly. In the following, to avoid the verification of the Razumikhin condition and to establish the continuous controller explicitly, we propose an alternative CLRF, which is based on the steepest descent feedback controller \cite{Clarke2010discontinuous, Pepe2017control}.

\begin{definition}
\label{def-3}
For the system \eqref{eqn-1}, a function $V\in\mathcal{C}(\mathbb{R}^{n}, \mathbb{R}^{+})$ is called a \textit{control Lyapunov-Razumikhin function (CLRF-II)}, if item (i) in Definition \ref{def-2} holds, and there exist $\gamma_{\rclf}, \eta_{\rclf}, \mu_{\rclf}\in\mathbb{R}^{+}$ such that $\gamma_{\rclf}>\eta_{\rclf}$, and for any nonzero $\phi\in\mathcal{C}([-\Delta, 0], \mathbb{R}^{n})$ with $\phi(0)=x$,
\begin{align}
\label{eqn-2}
&\inf_{u\in\mathbb{U}}\{L_{f}V(\phi)+L_{g}V(\phi)u\}<-\gamma_{\rclf}V(x)+\eta_{\rclf}\|e^{\mu_{\rclf}\theta}V(\phi)\|_{\dly},
\end{align}
where $\|e^{\mu_{\rclf}\theta}V(\phi)\|_{\dly}:=\sup_{\theta\in[-\Delta, 0]}e^{\mu_{\rclf}\theta}V(\phi(\theta))$.
\end{definition}

Different from Definition \ref{def-2} based on the Razumikhin condition, the Razumikhin condition is not needed in \eqref{eqn-2}, which will be applied to facilitate the controller design afterwards. The introduction of $\mu_{\rclf}\in\mathbb{R}^{+}$ is to increase the flexibility of \eqref{eqn-2}, and can be set as 0 simply. The following theorem shows that the CLRF-II is a CLRF-I under some reasonable conditions.

\begin{lemma}
\label{thm-1}
Consider the system \eqref{eqn-1} with a CLRF-II $V\in\mathcal{C}(\mathbb{R}^{n}, \mathbb{R}^{+})$, if there exists $\rho\in\mathcal{K}$ with $\rho(v(0))\geq\|v\|$ for all $v\in\mathcal{C}([-\Delta, 0], \mathbb{R}^{n})$ such that $\gamma_{\rclf}-\eta_{\rclf}\circ\rho\in\mathcal{K}$, then $V$ is a CLRF-I with $\gamma_{\cclf}=\gamma_{\rclf}-\eta_{\rclf}\circ\rho$ and $\rho_{\cclf}=\rho$.
\end{lemma}

Lemma \ref{thm-1} is similar to Remark 4 in \cite{Pepe2017control}, and the proof is omitted here. We emphasize that the converse is not necessarily valid. The reason lies in that when the Razumikhin condition is not satisfied, the evolution of the CLRF-I is unknown from Definition \ref{def-2}, whereas the evolution of the CLRF-II is bounded via \eqref{eqn-2}. With the CLRF-II, the SCP is extended to the time-delay case, which is presented below.

\begin{definition}
\label{def-4}
Consider the system \eqref{eqn-1} admitting the CLRF-II $V\in\mathcal{C}(\mathbb{R}^{n}, \mathbb{R}^{+})$, the system \eqref{eqn-1} is said to satisfy the \textit{Razumikhin-type small control property (R-SCP)}, if for arbitrary $\varepsilon>0$, there exists $\delta>0$ such that for any nonzero $\phi\in\mathcal{C}([-\Delta, 0], \mathbf{B}(\delta))$ with $\phi(0)=x$, there exists $u\in\mathbf{B}(\varepsilon)$ such that $L_{f}V(\phi)+L_{g}V(\phi)u<-\gamma_{\rclf}V(x)+\eta_{\rclf}\|e^{\mu_{\rclf}\theta}V(\phi)\|_{\dly}$.
\end{definition}

Following the Razumikhin approach, Definition \ref{def-4} extends the SCP into the time-delay case. Note that the R-SCP is satisfied for all nonzero $\phi\in\mathcal{C}([-\Delta, 0], \mathbf{B}(\delta))$ due to time delays. With the CLRF-II and the R-SCP, the closed-form controller is derived explicitly in the next theorem such that the stabilization objective is achieved for time-delay systems. This theorem extends the Sontag's formula in \cite{Sontag1989universal} to the time-delay case, and the proof is presented in Appendix \ref{asubsec-pf2}.

\begin{theorem}
\label{thm-2}
If the time-delay system \eqref{eqn-1} admits a CLRF-II $V\in\mathcal{C}(\mathbb{R}^{n}, \mathbb{R}^{+})$ and satisfies the R-SCP, then the controller $u(\phi):=\kappa(\lambda, \mathfrak{a}_{\rclf}(\phi), (L_{g}V(\phi))^{\top})$ defined as
\begin{equation}
\label{eqn-3}
\kappa(\lambda, p, q):=\left\{\begin{aligned}&\frac{p+\sqrt{p^{2}+\lambda\|q\|^{4}}}{-\|q\|^{2}}q, &\text{ if }& q\neq 0, \\
&0, &\text{ if }&  q=0, \end{aligned}\right.
\end{equation}
with  $\lambda>0$ and $\mathfrak{a}_{\rclf}(\phi):=L_{f}V(\phi)+\gamma_{\rclf}V(x)-\eta_{\rclf}\|e^{\mu_{\rclf}\theta}V(\phi)\|_{\dly}$, is continuous at the origin and ensures the GAS of the closed-loop system.
\end{theorem}

\subsection{Razumikhin-type Control Barrier Function}
\label{subsec-cbrf}

A commonly-used approach to investigate the safety specification is based on CBFs, which are generally defined via the safe set \cite{Ames2016control}. However, for physical systems like robotic systems, the workspace and unsafe set are known \textit{a priori}. Hence, an intuitive way to define the CBF is based on the unsafe set \cite{Wieland2007constructive}. To be specific, given the unsafe set $\mathbb{D}\subset\mathbb{R}^{n}$, a function $B\in\mathcal{C}(\mathbb{R}^{n}, \mathbb{R})$ is associated such that
\begin{align}
\label{eqn-4}
\mathbb{D}&\subseteq\{x\in\mathbb{R}^{n}: B(x)>0\}.
\end{align}
Moreover, with the function $B\in\mathcal{C}(\mathbb{R}^{n}, \mathbb{R})$, the Razumikhin-type CBF is proposed below for time-delay systems.

\begin{definition}
\label{def-5}
Consider the system \eqref{eqn-1} with the unsafe set $\mathbb{D}\subset\mathbb{R}^{n}$, a function $B\in\mathcal{C}(\mathbb{R}^{n}, \mathbb{R})$ is called a \textit{control barrier-Razumikhin function (CBRF)}, if there exist $\mu_{\rcbf}\geq0$ and $\gamma_{\rcbf}>\eta_{\rcbf}\geq0$ such that $\mathbb{S}_{\rcbf}:=\{x\in\mathbb{R}^{n}: B(x)\leq0\}\neq\varnothing$, and for any nonzero $\phi\in\mathcal{C}([-\Delta, 0], \mathbb{R}^{n}\setminus\mathbb{D})$ with $\phi(0)=x$,
\begin{align}
\label{eqn-5}
&\inf_{u\in\mathbb{U}}\{L_{f}B(\phi)+L_{g}B(\phi)u\}<-\gamma_{\rcbf}B(x)+\eta_{\rcbf}\|e^{\mu_{\rcbf}\theta}B(\phi)\|_{\dly}.
\end{align}
\end{definition}

In Definition \ref{def-5}, the non-emptiness of the set $\mathbb{S}_{\rcbf}$ is to guarantee the non-emptiness of the safe set and further the safety of the initial state, and the condition \eqref{eqn-5} is motivated from \eqref{eqn-2}. From the similarity between \eqref{eqn-5} and \eqref{eqn-2}, the following theorem presents the closed-form safety controller.

\begin{theorem}
\label{thm-3}
Given the unsafe set $\mathbb{D}\subset\mathbb{R}^{n}$, if the system \eqref{eqn-1} admits a CBRF $B\in\mathcal{C}(\mathbb{R}^{n}, \mathbb{R})$ and satisfies the R-SCP, then the controller $u(\phi):=\kappa(\lambda, \mathfrak{a}_{\rcbf}(\phi), (L_{g}B(\phi))^{\top})$ with the function $\kappa$ defined in \eqref{eqn-3} and $\mathfrak{a}_{\rcbf}(\phi):=L_{f}B(\phi)+\gamma_{\rcbf}B(x)-\eta_{\rcbf}\|e^{\mu_{\rcbf}\theta}B(\phi)\|_{\dly}$, is continuous at the origin and ensures the safety of the closed-loop system with the initial condition $\xi\in\mathcal{C}([-\Delta, 0], \mathbb{S}_{\rcbf})$.
\end{theorem}

The proof of Theorem \ref{thm-3} is presented in Appendix \ref{asubsec-pf3}. Note that from \eqref{eqn-4}, $B(x)>0$ may hold for some $x\in\mathbb{R}^{n}\setminus\mathbb{D}$. In this case, if the boundaries of $\mathbb{D}$ and $\mathbb{S}_{\rcbf}$ do not intersect such that $\mathbb{S}_{\rcbf}$ is entered first, that is, $\overline{\mathbb{R}^{n}\setminus(\mathbb{D}\cup\mathbb{S}_{\rcbf})}\cap\overbar{\mathbb{D}}=\varnothing$,  then the initial condition is allowed to be in $\mathcal{C}([-\Delta, 0], \mathbb{R}^{n}\setminus\mathbb{D})$. In particular, $\overline{\mathbb{R}^{n}\setminus(\mathbb{D}\cup\mathbb{S}_{\rcbf})}\cap\overbar{\mathbb{D}}=\varnothing$ implies $\partial(\mathbb{D}\cup\mathbb{S}_{\rcbf})\cap\partial\mathbb{D}=\varnothing$, and thus $(\mathbb{D}+\mathbf{B}(\varepsilon))\cap\partial\mathbb{D}\subset\mathbb{S}_{\rcbf}$ for arbitrarily small $\varepsilon>0$.

\subsection{Razumikhin-type Control Lyapunov-Barrier Function}
\label{subsec-combineA}

To incorporate the stabilization and safety objectives simultaneously, the optimization techniques have been extensively applied \cite{Jankovic2018robust, Ames2016control} to merge the CLF and CBF. However, it not easy to solve time-delay optimization problems to derive closed-form analytical solutions \cite{Wu2019new}. To avoid this issue and motivated by the existing results \cite{Romdlony2016stabilization}, a novel Razumikhin-type control function is proposed below to study the stabilization and safety objectives simultaneously.

\begin{definition}
\label{def-6}
Consider the system \eqref{eqn-1} with the unsafe set $\mathbb{D}\subset\mathbb{R}^{n}$, a proper and lower-bounded function $W\in\mathcal{C}(\mathbb{R}^{n}, \mathbb{R})$ is called a \textit{control Lyapunov-barrier-Razumikhin function (CLBRF)}, if
\begin{enumerate}[(i)]
  \item $W: \mathbb{R}^{n}\rightarrow\mathbb{R}$ is positive on the set $\mathbb{D}$;
  \item for any nonzero $\phi\in\mathcal{C}([-\Delta, 0], \mathbb{R}^{n}\setminus\mathbb{D})$, there exist $\gamma_{\rclbf}>\eta_{\rclbf}\geq0$ and $\mu_{\rclbf}\geq0$ such that $\inf_{u\in\mathbb{U}}\{L_{f}W(\phi)+L_{g}W(\phi)u\}<-\gamma_{\rclbf}W(x)+\eta_{\rclbf}\|e^{\mu_{\rclbf}\theta}W(\phi)\|_{\dly}$;
  \item $\mathbb{S}_{\rclbf}=\{x\in\mathbb{R}^{n}: W(x)\leq0\}\neq\varnothing$;
  \item $\overline{\mathbb{R}^{n}\setminus(\mathbb{D}\cup\mathbb{S}_{\rclbf})}\cap\overbar{\mathbb{D}}=\varnothing$.
\end{enumerate}
\end{definition}

With the proposed CLBRF and R-SCP, the next theorem shows the controller design to ensure the safety and semi-GAS of the system \eqref{eqn-1} simultaneously, and the proof is given in Appendix \ref{asubsec-pf4}.

\begin{theorem}
\label{thm-4}
Given the unsafe set $\mathbb{D}\subset\mathbb{R}^{n}$, if the system \eqref{eqn-1} admits a CLBRF $W\in\mathcal{C}(\mathbb{R}^{n}, \mathbb{R})$ and satisfies the R-SCP, then the controller $u(\phi)=\kappa(\lambda, \mathfrak{a}_{\rclbf}(\phi), (L_{g}W(\phi))^{\top})$ with the function $\kappa$ defined in \eqref{eqn-3}, $\lambda>0$ and $\mathfrak{a}_{\rclbf}(\phi):=L_{f}W(\phi)+\gamma_{\rclbf}W(x)-\eta_{\rclbf}\|e^{\mu_{\rclbf}\theta}W(\phi)\|_{\dly}$, is continuous at the origin and ensures both semi-GAS and safety of the closed-loop system with the initial condition $\xi\in\mathcal{C}([-\Delta, 0], \mathbb{R}^{n}\setminus\mathbb{D})$.
\end{theorem}

Theorem \ref{thm-4} implies the simultaneous satisfaction of both safety and stabilization objectives. Next, we show how to construct the CLBRF via the proposed CLRF-II and CBRF.

\begin{theorem}
\label{thm-5}
Given the unsafe set $\mathbb{D}\subset\mathbb{R}^{n}$, and assume that the system \eqref{eqn-1} admits a CLRF-II $V\in\mathcal{C}(\mathbb{R}^{n}, \mathbb{R}^{+})$ and a CBRF $B\in\mathcal{C}(\mathbb{R}^{n}, \mathbb{R})$, if
\begin{enumerate}[(i)]
  \item there exist $\mathbb{X}\subset\mathbb{R}^{n}\setminus\{0\}$ and a continuous function $\varphi: \mathbb{R}^{+}\rightarrow\mathbb{R}^{+}$ such that $\mathbb{D}\subset\mathbb{X}$ and $B(x)\leq-\varphi(\|x\|)$ for all $x\notin\mathbb{X}$;
  \item there exists $\psi\in\mathbb{R}^{+}$ such that $\alpha_{2}(\|x\|)<\psi\varphi(\|x\|)$ for all $x\in\partial\mathbb{X}$, where $\alpha_{2}$ is in Definition \ref{def-2};
  \item $\min\{\gamma_{\rclf}, \gamma_{\rcbf}\}>\max\{\eta_{\rclf}, \eta_{\rcbf}\}$,
\end{enumerate}
then $W(x)=V(x)+\psi B(x)$ is a CLBRF for the system \eqref{eqn-1} with the initial condition satisfying $\xi\in\mathcal{C}([-\Delta, 0], \mathbb{R}^{n}\setminus\mathbf{D})$ with $\mathbf{D}:=\{x\in\mathbb{X}: W(x)>0\}$.
\end{theorem}

The proof of Theorem \ref{thm-5} is presented in Appendix \ref{asubsec-pf5}. With the CLRF-II and CBRF, item (iii) can be verified easily. With item (i), the existence of $\psi\in\mathbb{R}^{+}$ in item (ii) can be verified. Item (i) can be satisfied via the construction and computation. Specifically, let $\mathbb{X}:=\mathbb{D}+\mathbf{B}(\varepsilon)$ with arbitrarily small $\varepsilon>0$, and for all $x\in\partial\mathbb{X}$, $B(x)$ can be computed and upper bounded via some function $-\varphi(\|x\|)$. Then, $B(x)$ can be set as $-\varphi(\|x\|)$ for all $x\notin\mathbb{X}$, which implies the satisfaction of item (i). Theorem \ref{thm-5} extends the results in \cite{Romdlony2016stabilization} into the time-delay case. Different from the exponentially stabilizing CLF and the constant property of the function $\varphi$ in \cite{Romdlony2016stabilization}, all conditions here are allowed to depend on the system state.

\section{Numerical Example}
\label{sec-examples}

Consider the following mechanical system from \cite{Romdlony2016stabilization}
\begin{align}
\label{eqn-6}
\dot{x}_{1}=x_{2}, \quad \dot{x}_{2}=-h(x_{t})-x_{1}+u,
\end{align}
where $x=(x_{1}, x_{2})\in\mathbb{R}^{2}$ with $x_{1}$ being the displacement and $x_{2}$ being the velocity, and $u\in\mathbb{R}$ is the control input. In \eqref{eqn-6}, $h(x_{t})=(0.8+2e^{-100|x_{2}(t-\tau)|})\tanh(10x_{2}(t-\tau))+x_{2}(t-\tau)$ is the delayed friction model to describe the damping parameter, where $\tau\in[0, 0.3]$ is the time delay. Therefore, the system \eqref{eqn-6} can be written as the form of \eqref{eqn-1} with  $f(x)=(x_{2}, -h(x_{t})-x_{1})$ and $g(x)=(0, 1)$.

To guarantee the stabilization objective, we define the Lyapunov candidate as $V(x):=x_{1}^{2}+x_{1} x_{2}+x_{2}^{2}$. By the detailed computation, we have that item (i) in Definition \ref{eqn-2} holds with $\alpha_{1}(v)=0.5v^{2}$ and $\alpha_{2}(v)=1.5v^{2}$ for all $v\geq0$. Hence, $V(x)$ is a CLRF-II if the condition \eqref{eqn-2} holds. For the safety objective, we assume the unsafe set to be $\mathbb{D}:=\{x\in\mathfrak{X}: \mathcal{H}(x)<4\}$ with $\mathcal{H}(x):=(1-(x_{1}+2)^{2})^{-1}+(1-(x_{2}-1)^{2})^{-1}$. Then, the corresponding barrier function is defined as
\begin{align}
B(x)=\left\{\begin{aligned}
&(e^{-\mathcal{H}(x)}-e^{-4})\|x\|^{2}, && \forall x\in\mathfrak{X},  \\
&-e^{-4}\|x\|^{2}, && \text{elsewhere},
\end{aligned}\right.
\end{align}
where $\mathfrak{X}:=(-3, -1)\times(0, 2)$. Obviously, $B(x)>0$ holds for $x\in\mathbb{D}$, and $\mathbb{S}_{\rcbf}:=\{x\in\mathbb{R}^{n}: B(x)\leq0\}\neq\varnothing$. Then $B(x)$ is a CBRF if the condition \eqref{eqn-5} holds.

Based on $V(x)$ and $B(x)$, we next construct the CLBRF $W(x):=V(x)+\psi B(x)$ to investigate the stabilization and safety objectives simultaneously. First, item (i) in Theorem \ref{thm-5} holds with $\varphi(v):=e^{-4}v^{2}$ for all $v\geq0$. From item (i) in Theorem \ref{thm-5}, $\psi>81.8972$. Item (iii) can be satisfied based on the construction of the CLRF $V(x)$ and CBRF $B(x)$. Therefore, $W(x)=V(x)+\psi B(x)$ is the CLBRF for the system \eqref{eqn-6}, and will be used to design the controller of the form \eqref{eqn-3} to guarantee the stabilization and safety objectives simultaneously. Given $\gamma_{\rclbf}=2.5, \eta_{\rclbf}=2$ and the gain $\lambda=2$, Fig. \ref{fig-1} shows the state trajectories of the closed-loop system starting from different initial conditions. From Fig. \ref{fig-1}, all trajectories converge to zero while avoiding the unsafe set $\mathbb{D}$, which thus verifies the efficiency of the proposed approach.

\begin{figure}
\centering
\begin{picture}(60, 95)
\put(-60, -8){\resizebox{60mm}{35mm}{\includegraphics[width=2.5in]{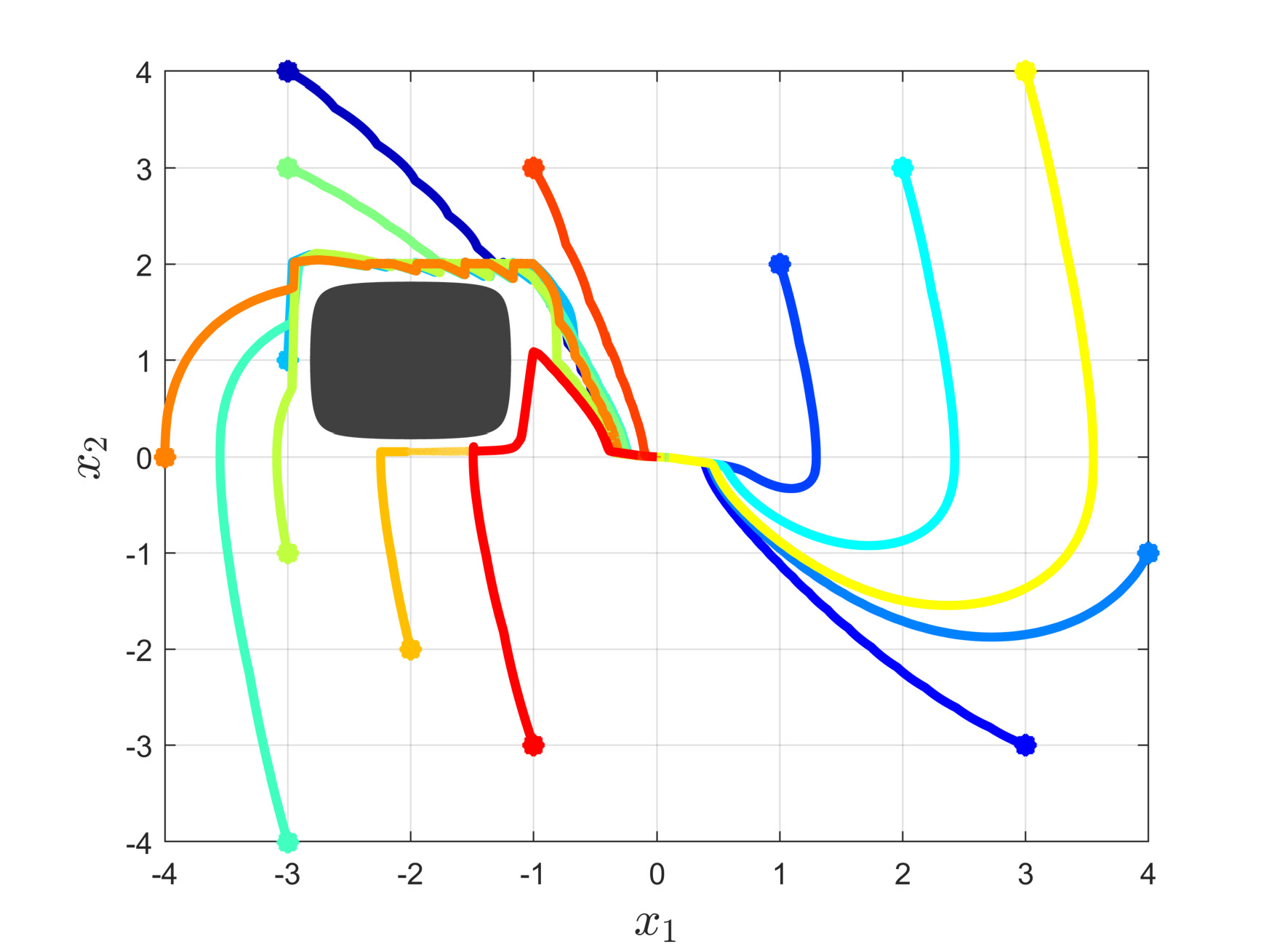}}}
\end{picture}
\caption{The simulation of the closed-loop system based on the proposed CLBRF. The dark grey region is the unsafe set $\mathbb{D}$, and the solid curves are the state trajectories starting from different constant initial conditions.}
\label{fig-1}
\end{figure}

\section{Conclusion}
\label{sec-conclusion}

This paper provided a novel framework for the control design of safety-critical systems with time delays. Based on the Razumikhin approach, the Razumikhin-type control Lyapunov and barrier functions were proposed to investigate the stabilization and safety control problems. To achieve the safety and stabilization objectives simultaneously, the proposed Razumikhin-type control functions were combined such that the stabilizing and safety controllers can be merged. Future work will be devoted to decentralized safety control for multi-agent systems with time delays.

\appendix
\setcounter{equation}{0}
\renewcommand{\theequation}{A.\arabic{equation}}

\subsection{Technical Lemma}

\begin{lemma}
\label{lem-A1}
Given a locally Lipschitz function $V: \mathbb{R}^{n}\rightarrow\mathbb{R}^{+}$, if there exist $\gamma>\eta>0$ such that for all $t\in\mathbb{R}^{+}$,
\begin{align*}
D^{+}V(x(t))\leq-\gamma V(x(t))+\eta\sup_{\theta\in[-\Delta, 0]} e^{\mu\theta}V(x(t+\theta)),
\end{align*}
where $D^{+}V(x(t))= \lim\sup_{s\rightarrow0^{+}}(V(t+s)-V(t))/s$ is the upper Dini derivative of $V(t)=V(x(t))$, then there exists $\varrho\in(0, \bar{\varrho})$ such that $V(x(t))\leq e^{-\varrho t}V(0)$ for all $t>0$, where $\bar{\varrho}$ is the solution to $2\varrho-\gamma+\eta e^{\Delta\varrho}=0$.
\end{lemma}

\begin{proof}
First, define a function $\Gamma(\varrho):=\varrho-\gamma+\eta e^{\Delta\varrho}$. Since $\gamma>\eta>0$, $\Gamma(0)=-\gamma+\eta<0$, and $\Gamma(\varrho)\rightarrow\infty$ as $\varrho\rightarrow\infty$. $\Gamma'(\varrho):=1+\eta\Delta e^{\Delta\varrho}>0$. Thus, there exists a unique $\bar{\varrho}>0$ such that $\Gamma(\bar{\varrho})=0$, and $\Gamma(\varrho)<0$ for all $\varrho\in(0, \bar{\varrho})$.

Next, we prove that for any $\varrho\in(0, \bar{\varrho})$,
\begin{align}
\label{eqn-A2}
V(t)&\leq e^{-\varrho t}V(0)=:\mathfrak{U}(t), \quad \forall t\geq0.
\end{align}
Obviously, \eqref{eqn-A2} holds at $t=0$. We claim that \eqref{eqn-A2} is valid for all $t\in\mathbb{R}^{+}$. If not, we assume that $t^{\ast}:=\inf\{t\geq0: V(t)>\mathfrak{U}(t)\}$ is the first time instant such that
\begin{align}
\label{eqn-A3}
V(t^{\ast})&=\mathfrak{U}(t^{\ast}),  \quad   V(t)\leq\mathfrak{U}(t), \quad \forall 0<t<t^{\ast},  \\
\label{eqn-A3.1}
V(t)&>\mathfrak{U}(t), \quad \forall t\in(t^{\ast}, t^{\ast}+\Delta t),
\end{align}
where $\Delta t>0$ is arbitrarily small. From \eqref{eqn-A3}-\eqref{eqn-A3.1},
\begin{align}
\label{eqn-A4}
D^{+}V(t^{\ast})&\geq D^{+}\mathfrak{U}(t^{\ast}).
\end{align}
However, from the definition of $\bar{\varrho}$ and \eqref{eqn-A3},
\begin{align}
\label{eqn-A4.1}
D^{+}\mathfrak{U}(t^{\ast})&>(-\gamma+\eta e^{\varrho\Delta})e^{-\varrho t^{\ast}}V(0)  \nonumber  \\
&\geq-\gamma V(t^{\ast})+\eta\sup_{-\Delta\leq\theta\leq0}e^{(\Delta+\theta)\varrho}V(t^{\ast}+s) \nonumber  \\
&\geq-\gamma V(t^{\ast})+\eta\sup_{-\Delta\leq\theta\leq0}e^{\varrho\Delta+(\varrho-\mu)\theta}e^{\mu\theta}V(t^{\ast}+\theta)  \nonumber  \\
&\geq D^{+}V(t^{\ast}),
\end{align}
where the first ``$\geq$'' holds from \eqref{eqn-A3}, and last ``$\geq$'' holds from the fact that $e^{\varrho\Delta+(\varrho-\mu)\theta}>1$. Since \eqref{eqn-A4.1} contradicts with \eqref{eqn-A4}, the claim holds for all $t>0$.
\end{proof}

%

\subsection{Proof of Theorem \ref{thm-2}}
\label{asubsec-pf2}

First, we show the stabilization of the system \eqref{eqn-1} under the designed controller. If $L_{g}V(\phi)\equiv0$, then $u\equiv0$, and from \eqref{eqn-2} in Definition \ref{def-3}, we have
\begin{align}
\label{eqn-A6}
&L_{f}V(\phi)+L_{g}V(\phi)u+\gamma_{\rclf}V(x)-\eta_{\rclf}\|e^{\mu_{\rclf}\theta}V(\phi)\|_{\dly}  \nonumber  \\
&=L_{f}V(\phi)+\gamma_{\rclf}V(x)-\eta_{\rclf}\|e^{\mu_{\rclf}\theta}V(\phi)\|_{\dly}<0.
\end{align}
For the case $L_{g}V(\phi)\neq0$, we define $\mathfrak{a}_{\rclf}(\phi):=L_{f}V(\phi)+\gamma_{\rclf}V(x)-\eta_{\rclf}\|e^{\mu_{\rclf}\theta}V(\phi)\|_{\dly}$ and $\mathfrak{b}_{\rclf}(\phi):=L_{g}V(\phi)$ to simplify the notation. From the designed control law, we have
\begin{align}
\label{eqn-A7}
&L_{f}V(\phi)+L_{g}V(\phi)u+\gamma_{\rclf}V(x)-\eta_{\rclf}\|e^{\mu_{\rclf}\theta}V(\phi)\|_{\dly}   \nonumber  \\
&=-\sqrt{\mathfrak{a}^{2}_{\rclf}(\phi)+\lambda\|\mathfrak{b}_{\rclf}(\phi)\|^{4}}\leq0.
\end{align}

From \eqref{eqn-A6}-\eqref{eqn-A7}, we have that for all $t\in\mathbb{R}^{+}$,
\begin{align*}
D^{+}V(x(t))\leq-\gamma_{\rclf}V(x(t))+\eta_{\rclf}\|e^{\mu_{\rclf}\theta}V(\phi)\|_{\dly},
\end{align*}
where the Dini derivative is the usual derivative due to $V\in\mathcal{C}(\mathbb{R}^{n}, \mathbb{R}^{+})$. From Lemma \ref{lem-A1}, there exists $\varrho_{\rclf}>0$ such that $V(x(t))\leq e^{-\varrho_{\rclf}t}V(0)$ for all $t>0$, combining which with $V(0)\leq\|V(\xi)\|_{\dly}$ and item (i) in Definition \ref{def-2}, we conclude that the stabilization objective is achieved for the system \eqref{eqn-1}.

Second, we show the continuity of the controller at the origin. From the R-SCP, for arbitrary $\varepsilon\in\mathbb{R}^{+}$, there exists $\delta\in\mathbb{R}^{+}$ such that for any nonzero $\phi\in\mathcal{C}([-\Delta, 0], \mathbf{B}(\delta))$, there exists $u\in\mathbf{B}(\varepsilon)$ such that $L_{f}V(\phi)+L_{g}V(\phi)u<-\gamma_{\rclf}V(x)+\eta_{\rclf}\|e^{\mu_{\rclf}\theta}V(\phi)\|_{\dly}$, which implies $L_{f}V(\phi)+\gamma_{\rclf}V(x)\leq\varepsilon\|L_{g}V(\phi)\|+\eta_{\rclf}\|e^{\mu_{\rclf}\theta}V(\phi)\|_{\dly}$. Since $V$ is continuously differentiable and $g$ in \eqref{eqn-1} is locally Lipschitz, there exists $\hat{\delta}\in\mathbb{R}^{+}$ with $\hat{\delta}\neq\delta$ such that $\|L_{g}V(\phi)\|\leq\varepsilon$ for any nonzero $\phi\in\mathcal{C}([-\Delta, 0], \mathbf{B}(\hat{\delta}))$. Let $\delta_{\min}:=\min\{\hat{\delta}, \delta\}$, and for any nonzero $\phi\in\mathcal{C}([-\Delta, 0], \mathbf{B}(\delta_{\min}))$,
\begin{align*}
L_{f}V(\phi)+\gamma_{\rclf}V(x)&\leq\varepsilon^{2}+\eta_{\rclf}\|e^{\mu_{\rclf}\theta}V(\phi)\|_{\dly} \\
&\leq\varepsilon^{2}+\eta_{\rclf}\alpha_{2}(\delta_{\min}):=\bar{\varepsilon}^{2},  \\
\|u\|\leq\frac{\bar{\varepsilon}^{2}+\sqrt{\bar{\varepsilon}^{4}+\lambda\varepsilon^{4}}}{\varepsilon}&\leq2\bar{\varepsilon}^{2}/\varepsilon+\sqrt{\lambda}\varepsilon:=\hat{\varepsilon}.
\end{align*}
Hence, there exists $\delta_{\min}\in\mathbb{R}^{+}$ such that for any nonzero $\phi\in\mathcal{C}([-\Delta, 0], \mathbf{B}(\delta_{\min}))$, there exists $u\in\mathbf{B}(\hat{\varepsilon})$ such that $L_{f}V(\phi)+L_{g}V(\phi)u<-\gamma_{\rclf}V(x)+\eta_{\rclf}\|e^{\mu_{\rclf}\theta}V(\phi)\|_{\dly}$. Hence, the controller is continuous at the origin since $\delta_{\min}, \hat{\varepsilon}\in\mathbb{R}^{+}$ can be chosen arbitrarily small.

\subsection{Proof of Theorem \ref{thm-3}}
\label{asubsec-pf3}

If $L_{g}B(\phi)\equiv0$, then $u\equiv0$, and we have
\begin{align}
\label{eqn-A8}
&L_{f}B(\phi)+L_{g}B(\phi)u+\gamma_{\rcbf}B(x)-\eta_{\rcbf}\|e^{\mu_{\rcbf}\theta}B(\phi)\|_{\dly}  \nonumber\\
&=L_{f}B(\phi)+\gamma_{\rcbf}B(x)-\eta_{\rcbf}\|e^{\mu_{\rcbf}\theta}B(\phi)\|_{\dly}<0.
\end{align}
If $L_{g}B(\phi)\neq0$, then from the designed controller,
\begin{align}
\label{eqn-A9}
&L_{f}B(\phi)+L_{g}B(\phi)u+\gamma_{\rcbf}B(x)-\eta_{\rcbf}\|e^{\mu_{\rcbf}\theta}B(\phi)\|_{\dly} \nonumber\\
&=-\sqrt{\bar{\mathfrak{a}}^{2}_{\rcbf}(\phi)+\lambda\|\bar{\mathfrak{b}}_{\rcbf}(\phi)\|^{4}}\leq0.
\end{align}
where $\bar{\mathfrak{a}}_{\rcbf}(\phi)=L_{f}B(\phi)+\gamma_{\rcbf}B(x)-\eta_{\rcbf}\|e^{\mu_{\rcbf}\theta}B(\phi)\|_{\dly}$, and $\bar{\mathfrak{b}}_{\rcbf}(\phi)=L_{g}B(\phi)$. Combining \eqref{eqn-A8} and \eqref{eqn-A9} yields that for any nonzero $\phi\in\mathcal{C}([-\Delta, 0], \mathbb{S}_{\rcbf})$,
\begin{align*}
&L_{f}B(\phi)+L_{g}B(\phi)u\leq-\gamma_{\rcbf}B(x)+\eta_{\rcbf}\|e^{\mu_{\rcbf}\theta}B(\phi)\|_{\dly}.
\end{align*}
From Lemma \ref{lem-A1}, $B$ decreases along the trajectory of the system \eqref{eqn-1} with the designed controller. For the initial condition $\xi\in\mathcal{C}([-\Delta, 0], \mathbb{S}_{\rcbf})$, $B(\xi(\theta))\leq0$ for all $\theta\in[-\Delta, 0]$, which implies that $B(x(t))\leq e^{-\varrho_{\rcbf}t}B(x(0))<0$ for all $t\in\mathbb{R}^{+}$. Hence, $\mathbb{S}_{\rcbf}$ is forward invariant and the safety objective is achieved for the system \eqref{eqn-1}.

\subsection{Proof of Theorem \ref{thm-4}}
\label{asubsec-pf4}

First, we consider the case $\xi\in\mathcal{C}([-\Delta, 0], \mathbb{S}_{\rclbf})$, we have that $W(\xi(\theta))\leq0$ for all $\theta\in[-\Delta, 0]$. Similar to the proof of Theorem \ref{thm-2}, we yield that $D^{+}W(x(t))\leq-\gamma_{\rclbf}W(x(t))+\eta_{\rclbf}\|e^{\mu_{\rclbf}\theta}W(\phi)\|_{\dly}$, and thus $W(x(t))\leq e^{-\varrho_{\rclbf}t}W(x(0))$ for all $t\in\mathbb{R}^{+}$. We conclude that $W(x(t))<0$ for all $t\in\mathbb{R}^{+}$, which implies that the set $\mathbb{S}_{\rclbf}$ is forward invariant and $x(t)\notin\mathbb{D}$ for all $t\in\mathbb{R}^{+}$. In addition, the compactness of the set $\mathbb{S}_{\rclbf}$ comes from the properness of the function $W(x(t))$, and further we have that $\lim_{t\rightarrow\infty}x(t)\notin\mathbb{D}$.

Next, consider the case $\xi\in\mathcal{C}([-\Delta, 0], \mathbb{R}^{n}\setminus(\mathbb{D}\cup\mathbb{S}_{\rclbf}))$. Obviously, $\xi(\theta)\notin\mathbb{D}$ for all $\theta\in[-\Delta, 0]$. In addition, $W(\xi(\theta))\geq0$, and $W$ is non-increasing along the system trajectory. The condition $\overline{\mathbb{R}^{n}\setminus(\mathbb{D}\cup\mathbb{S}_{\rclbf})}\cap\overbar{\mathbb{D}}=\varnothing$ implies that $\overline{\mathbb{R}^{n}\setminus(\mathbb{D}\cup\mathbb{S}_{\rclbf})}$ does not intersect with $\overbar{\mathbb{D}}$. Therefore, the trajectory starting from $\mathbb{R}^{n}\setminus(\mathbb{D}\cup\mathbb{S}_{\rclbf})$ will reach the boundary of $\mathbb{R}^{n}\setminus(\mathbb{D}\cup\mathbb{S}_{\rclbf})$ instead of enter into the set $\mathbb{D}$ directly. Note that $W(x)=0$ holds on the boundary of $\mathbb{R}^{n}\setminus(\mathbb{D}\cup\mathbb{S}_{\rclbf})$, combining which with $W(x(t))\leq e^{-\varrho_{\rclbf}t}W(\xi(0))$ yields that $W(x(t))\leq W(\xi(0))$ for all $t\in\mathbb{R}^{+}$, which implies that $x(t)$ will remain in $\mathbb{S}_{\rclbf}$ for all $t\in\mathbb{R}^{+}$.

Finally, we prove that the system \eqref{eqn-1} is stabilizable under the designed controller. Let $\xi\in\mathcal{C}([-\Delta, 0], \mathbb{R}^{n}\setminus\mathbb{D})$, and from the above analysis, $x(t)\notin\mathbb{D}$ for all $t\in\mathbb{R}^{+}$. Following the similar fashion of the proof of Theorem \ref{thm-2}, $L_{f}W(\phi)+L_{g}W(\phi)u\leq-\gamma_{\rclbf}W(x)
+\eta_{\rclbf}\|e^{\mu_{\rclbf}\theta}W(\phi)\|_{\dly}$ for any nonzero $x\in\mathbb{R}^{n}\setminus\mathbb{D}$, combining which with Lemma \ref{lem-A1} yields
\begin{align}
\label{eqn-A10}
W(x(t))\leq e^{-\varrho_{\rclbf}t}W(\xi(0)), \quad \forall t\in\mathbb{R}^{+}.
\end{align}
Since the function $W$ is proper and lower-bounded, \eqref{eqn-A10} shows that the trajectory is bounded and that the closure of the set $\{x(t): t\in[0, \infty)\}$ is compact. This implies that the $\omega$-limit set $\Omega(x):=\{\mathbf{x}\in\mathbb{R}^{n}: \exists\{t_{k}\in\mathbb{R}^{+}: k\in\mathbb{N}^{+}\} \text{ such that }\lim_{t_{k}\rightarrow\infty}x(t_{k})=\mathbf{x}\}$ is non-empty, compact, connected and invariant. From \eqref{eqn-A10}, the function $W(x(t))$ is monotonically decreasing with respect to $t\in\mathbb{R}^{+}$ and has its limit (which is denoted by $\mathbf{w}\in\mathbb{R}$) as $t\rightarrow\infty$. That is, $W(\mathbf{x})=\lim_{t_{k}\rightarrow\infty}W(x(t_{k}))=\mathbf{w}$. Hence, the set $\Omega(x)$ only contains one element, which further implies that $\Omega(W)$ only contains one element (i.e., $\mathbf{w}\in\mathbb{R}$). From \eqref{eqn-A10} and the safety guarantee, $\mathbb{D}\cap\Omega(\xi)=\varnothing$ and thus $W(x)$ converges to the origin, which implies that the system \eqref{eqn-1} is semi-GAS under the designed controller.

\subsection{Proof of Theorem \ref{thm-5}}
\label{asubsec-pf5}

Since $V$ and $B$ are continuously differentiable, the function $W$ is continuously differentiable. From the definition of $W$, we have for all $x\in\mathbb{D}$,
\begin{align*}
W(x)&=V(x)+\psi B(x)>V(x)\geq\alpha_{1}(\|x\|),
\end{align*}
where ``$>$'' holds because $B(x)>0$ for all $x\in\mathbb{D}$. From item (i) of Definition \ref{def-2}, $W(x)>0$ holds for all $x\in\mathbb{D}$. From the definition of $\mathbf{D}$, $W(x)>0$ for all $x\in\mathbf{D}$. For all $x\in\partial\mathbb{X}$, we have from item (i) of Definition \ref{def-2} and item (ii) of Theorem \ref{thm-5} that $W(x)\leq V(x)-\psi\varphi(\|x\|)\leq\alpha_{2}(\|x\|)-\psi\varphi(\|x\|)<0$, which implies the non-emptiness of the set $\mathbb{S}_{\rclbf}$.

Second, from the definition of the set $\mathbf{D}$, $\mathbb{D}\subset\mathbf{D}$,  $W(x)>0$ for all $x\in\mathbf{D}$, and $W(x)<0$ for all $x\in\partial\mathbb{X}$. In addition,  $\mathbf{D}\subseteq\inte(\mathbb{X})$, and $\partial\mathbb{X}\cap\partial\mathbf{D}=\varnothing$. From the definitions of the functions $V, B$, we have that for any nonzero $\phi\in\mathcal{C}([-\Delta, 0], \mathbb{R}^{n}\setminus\mathbf{D})$,
\begin{align*}
&\inf_{u\in\mathbb{U}}\{L_{f}W(\phi)+L_{g}W(\phi)u\} \nonumber\\
&\leq\inf_{u\in\mathbb{U}}\{L_{f}V(\phi)+\psi L_{f}B(\phi)+(L_{g}V(\phi)+\psi L_{g}B(\phi))u\}   \nonumber\\
&<-\lambda_{\rclbf}W(x)+\eta_{\rclbf}\|e^{\mu_{\rclbf}\theta}W(\phi)\|_{\dly},
\end{align*}
where $\lambda_{\rclbf}:=\min\{\lambda_{\rclf}, \lambda_{\rcbf}\}, \eta_{\rclbf}:=\max\{\eta_{\rclf}, \eta_{\rcbf}\}$ and $\mu_{\rclbf}:=\min\{\mu_{\rclf}, \mu_{\rcbf}\}$. Note that $\lambda_{\rclbf}>\eta_{\rclbf}$ holds from item (iii).

Finally, since $\partial\mathbb{X}\cap\partial\mathbf{D}=\varnothing$ and $\mathbf{D}\subset\mathbb{X}$, $\overbar{\mathbb{R}^{n}\setminus(\mathbf{D}\cup\mathbb{S})}\cap\overbar{\mathbf{D}}=\varnothing$, which implies the last item in Definition \ref{def-6}. After verifying all conditions in Definition \ref{def-6}, we conclude that the function $W$ is a CLBRF for the system \eqref{eqn-1} with the initial condition $\xi\in\mathcal{C}([-\Delta, 0], \mathbb{R}^{n}\setminus\mathbf{D})$.



\begin{thebibliography}{10}

\bibitem{Humayed2017cyber}
A.~Humayed, J.~Lin, F.~Li, and B.~Luo, ``Cyber-physical systems security-{A}
  survey,'' \emph{IEEE Internet Things J.}, vol.~4, no.~6, pp. 1802--1831,
  2017.

\bibitem{Sontag1989universal}
E.~D. Sontag, ``A `universal' construction of {A}rtstein's theorem on nonlinear
  stabilization,'' \emph{Syst. Control Lett.}, vol.~13, no.~2, pp. 117--123,
  1989.

\bibitem{Jankovic2001control}
M.~Jankovic, ``Control {L}yapunov-{R}azumikhin functions and robust
  stabilization of time delay systems,'' \emph{IEEE Trans. Autom. Control},
  vol.~46, no.~7, pp. 1048--1060, 2001.

\bibitem{Pepe2017control}
P.~Pepe, ``On control {L}yapunov-{R}azumikhin functions, nonconstant delays,
  nonsmooth feedbacks, and nonlinear sampled-data stabilization,'' \emph{IEEE
  Trans. Autom. Control}, vol.~62, no.~11, pp. 5604--5619, 2017.

\bibitem{Ames2016control}
A.~D. Ames, X.~Xu, J.~W. Grizzle, and P.~Tabuada, ``Control barrier function
  based quadratic programs for safety critical systems,'' \emph{IEEE Trans.
  Autom. Control}, vol.~62, no.~8, pp. 3861--3876, 2016.

\bibitem{Romdlony2016stabilization}
M.~Z. Romdlony and B.~Jayawardhana, ``Stabilization with guaranteed safety
  using control {L}yapunov-barrier function,'' \emph{Automatica}, vol.~66, pp.
  39--47, 2016.

\bibitem{Ogren2001control}
P.~Ogren, M.~Egerstedt, and X.~Hu, ``A control {L}yapunov function approach to
  multi-agent coordination,'' in \emph{Proc. IEEE Conf. Decis. Control}, IEEE, 2001, pp. 1150--1155.

\bibitem{Panagou2015distributed}
D.~Panagou, D.~M. Stipanovi{\'c}, and P.~G. Voulgaris, ``Distributed
  coordination control for multi-robot networks using {L}yapunov-like barrier
  functions,'' \emph{IEEE Trans. Autom. Control}, vol.~61, no.~3, pp. 617--632,
  2015.

\bibitem{Lindemann2018control}
L.~Lindemann and D.~V. Dimarogonas, ``Control barrier functions for signal
  temporal logic tasks,'' \emph{IEEE Control Syst. Lett.}, vol.~3, no.~1, pp.
  96--101, 2018.

\bibitem{Ngo2005integrator}
K.~B. Ngo, R.~Mahony, and Z.-P. Jiang, ``Integrator backstepping using barrier
  functions for systems with multiple state constraints,'' in \emph{Proc. IEEE
  Conf. Decis. Control}, IEEE, 2005, pp.
  8306--8312.

\bibitem{Jankovic2018robust}
M.~Jankovic, ``Robust control barrier functions for constrained stabilization
  of nonlinear systems,'' \emph{Automatica}, vol.~96, pp. 359--367, 2018.

\bibitem{Sipahi2011stability}
R.~Sipahi, S.-I. Niculescu, C.~T. Abdallah, W.~Michiels, and K.~Gu, ``Stability
  and stabilization of systems with time delay,'' \emph{IEEE Control Syst.
  Mag.}, vol.~31, no.~1, pp. 38--65, 2011.

\bibitem{Gao2019stability}
Q.~Gao and H.~R. Karimi, \emph{Stability, Control and Application of Time-Delay
  Systems}, Butterworth-Heinemann, 2019.

\bibitem{Fridman2014tutorial}
E.~Fridman, ``Tutorial on {L}yapunov-based methods for time-delay systems,''
  \emph{Eur. J. Control}, vol.~20, no.~6, pp. 271--283, 2014.

\bibitem{Ren2019krasovskii}
W.~Ren and J.~Xiong, ``Krasovskii and {R}azumikhin stability theorems for
  stochastic switched nonlinear time-delay systems,'' \emph{SIAM J. Control
  Optim.}, vol.~57, no.~2, pp. 1043--1067, 2019.

\bibitem{Teel1998connections}
A.~R. Teel, ``Connections between {R}azumikhin-type theorems and the {ISS}
  nonlinear small gain theorem,'' \emph{IEEE Trans. Autom. Control}, vol.~43,
  no.~7, pp. 960--964, 1998.

\bibitem{Ren2018vector}
W.~Ren and J.~Xiong, ``Vector-{L}yapunov-function-based input-to-state
  stability of stochastic impulsive switched time-delay systems,'' \emph{IEEE
  Trans. Autom. Control}, vol.~64, no.~2, pp. 654--669, 2018.

\bibitem{Prajna2005methods}
S.~Prajna and A.~Jadbabaie, ``Methods for safety verification of time-delay
  systems,'' in \emph{Proc. IEEE Conf. Decis. Control}.\hskip 1em plus 0.5em
  minus 0.4em\relax IEEE, 2005, pp. 4348--4353.

\bibitem{Jankovic2018control}
M.~Jankovic, ``Control barrier functions for constrained control of linear
  systems with input delay,'' in \emph{American Control Conference}.\hskip 1em
  plus 0.5em minus 0.4em\relax IEEE, 2018, pp. 3316--3321.

\bibitem{Pepe2014stabilization}
P.~Pepe, ``Stabilization in the sample-and-hold sense of nonlinear retarded
  systems,'' \emph{SIAM J. Control Optim.}, vol.~52, no.~5, pp. 3053--3077,
  2014.

\bibitem{Hale1993introduction}
J.~K. Hale and S.~M.~V. Lunel, \emph{Introduction to Functional Differential
  Equations}.\hskip 1em plus 0.5em minus 0.4em\relax Springer Science \&
  Business Media, 1993.

\bibitem{Sastry2013nonlinear}
S.~Sastry, \emph{Nonlinear Systems: Analysis, Stability, and Control}.\hskip
  1em plus 0.5em minus 0.4em\relax Springer Science \& Business Media, 2013,
  vol.~10.

\bibitem{Sepulchre2012constructive}
R.~Sepulchre, M.~Jankovic, and P.~V. Kokotovic, \emph{Constructive Nonlinear
  Control}.\hskip 1em plus 0.5em minus 0.4em\relax Springer Science \& Business
  Media, 2012.

\bibitem{Clarke2010discontinuous}
F.~Clarke, ``Discontinuous feedback and nonlinear systems,'' \emph{IFAC
  Proceedings}, vol.~43, no.~14, pp. 1--29, 2010.

\bibitem{Wieland2007constructive}
P.~Wieland and F.~Allg{\"o}wer, ``Constructive safety using control barrier
  functions,'' \emph{IFAC Proceedings}, vol.~40, no.~12, pp. 462--467, 2007.

\bibitem{Wu2019new}
D.~Wu, Y.~Bai, and C.~Yu, ``A new computational approach for optimal control
  problems with multiple time-delay,'' \emph{Automatica}, vol. 101, pp.
  388--395, 2019.
\end{thebibliography}
\end{document}